\documentclass[a4paper,10pt]{article}
\usepackage[utf8]{inputenc}
\usepackage{amsthm}
\usepackage{amsmath}
\usepackage{authblk}
\usepackage{cite}
\usepackage{textcomp}
\usepackage{amssymb}
\usepackage{hyperref}
\usepackage{relsize}
\newtheorem{theorem}{Theorem}

\newtheorem{lemma}{Lemma}

\title{Stochastic recursive inclusion in two timescales with an
application to the Lagrangian dual problem}
 \author[1]{
Arun Selvan. R
}
 \author[2]{
 Shalabh Bhatnagar
 }
\affil[1]{\texttt{arunselvan@csa.iisc.ernet.in}}
\affil[2]{\texttt{shalabh@csa.iisc.ernet.in}}
\affil[1,2]{Department of Computer Science and Automation,
 Indian Institute of Science,
 Bangalore - 560012, India.}
\begin{document}

\date{}
\maketitle
\begin{abstract}
In this paper we present a framework to analyze the asymptotic behavior of two timescale stochastic approximation
 algorithms including those with set-valued mean fields. This paper builds on the works of 
 Borkar and Perkins \& Leslie.
 The framework presented herein is more general as compared to the 
 synchronous two timescale framework of
 Perkins \& Leslie, however the assumptions involved are easily verifiable.
 As an application, we use this framework to analyze the
 two timescale stochastic approximation algorithm corresponding to the
 Lagrangian dual problem in optimization theory.

\end{abstract}

\section{Introduction}
The classical dynamical systems approach was developed by Bena\"{i}m \cite{Benaim96, Benaim99} and 
Bena\"{i}m and Hirsch \cite{BenaimHirsch}. They showed that the asymptotic behavior of a stochastic
approximation algorithm $(SA)$ can be studied by analyzing the asymptotics of the associated ordinary 
differential equation $(o.d.e.)$. This method is popularly known as the $o.d.e.$ \textit{method} and was 
originally introduced by Ljung \cite{Ljung77}.
In 2005, Bena\"{i}m, Hofbauer and Sorin \cite{Benaim05} extended the dynamical systems approach to include
the situation where the stochastic approximation algorithm tracks a solution to the associated
differential inclusion. Such algorithms are called \textit{stochastic recursive inclusions}.
For a detailed exposition on $SA$, the reader is referred to books by Borkar \cite{BorkarBook}
and Kushner and Yin \cite{KushnerYin}.
\paragraph{}
There are many applications where the aforementioned paradigms are inadequate. For example,
the right hand side of a $SA$ may require further averaging or an additional 
recursion to evaluate it.
An instance mentioned in Borkar \cite{Borkartt} is the `adaptive heuristic
critic' approach to reinforcement learning \cite{keerthi} that requires a stationary value iteration 
executed between two policy iterations. To solve such problems,
Borkar \cite{Borkartt} analyzed the two timescale $SA$ algorithms.
The two timescale paradigm presented in Borkar \cite{Borkartt} is inadequate
if the coupled iterates are stochastic recursive inclusions. 
Such iterates arise naturally in many learning algorithms, see for instance \textit{Section 5} of
\cite{Perkins}. For another application from convex optimization the reader is referred to
Section~\ref{Lagrange} of this paper.
Such iterates also arise in applications that involve projections onto non-convex sets.
The first attempt
at tackling this problem was made by Perkins and Leslie \cite{Perkins} in 2012.
They extended the two timescale scheme of Borkar \cite{Borkartt}
to include the situation when the two iterates track solutions to
differential inclusions.
\paragraph{}
Consider the following coupled recursion:
\begin{equation}\label{twotimescale}
\begin{split}
x_{n+1} = x_{n} + a(n) \left[ u_{n} + M^{1}_{n+1} \right], \\
y_{n+1} = y_{n} + b(n) \left[ v_{n} + M^{2}_{n+1} \right],
\end{split}
\end{equation}
where
$u_{n} \in h(x_{n}, y_{n})$, $v_{n} \in g(x_{n}, y_{n})$, $h: \mathbb{R}^{d+k} \to 
\left\{ subsets \ of\ \mathbb{R}^{d} \right\}$ and $g: \mathbb{R}^{d+k} \to \left\{ subsets\ of\ \mathbb{R}^{k} \right\}$.
Such iterates were analyzed in \cite{Perkins}. Further,
as an application a Markov decision process (MDP) based actor critic 
type learning algorithm was also presented in \cite{Perkins}.
\paragraph{}
In this paper we generalize the synchronous two timescale stochastic approximation
scheme presented in \cite{Perkins}.
We present sufficient conditions that are mild and easily verifiable.
For a complete list of assumptions
used herein, the reader is referred to Section~\ref{assumptions} and for the analyses
under these conditions
the reader is referred to Section~\ref{convergenceproof}.
It is worth noting that 
the analysis of the faster timescale proceeds in a predictable manner,
however,
 the \textit{analysis of the slower timescale presented herein is new
to the literature to the best of our knowledge}.
\paragraph{}
In convex optimization, one is interested in minimizing an objective function (that
is convex) subject to a few constraints. A solution to this optimization problem
is a set of vectors that minimize our objective function. 
Often this set is referred to as a minimum set. 
In Section~\ref{Lagrange}, we analyze the two timescale $SA$ algorithm
corresponding to the Lagrangian dual of a primal problem. 
As we shall see later, this analysis considers a family of minimum sets and as a 
consequence of our framework these minimum sets are no longer required to be \textit{singleton}.
In \cite{dantzig}, Dantzig, Folkman and Shapiro  presented sufficient conditions for the continuity
of minimum sets of continuous functions. We shall use results from that paper to show that
under some standard convexity conditions the assumptions of Section~\ref{assumptions}
are satisfied. We then conclude from our main result, Theorem~\ref{main}, that
the two timescale algorithm in question converges to a solution to the dual problem.
\section{Preliminaries and assumptions}
\subsection{Definitions and notations} \label{definition}
The definitions and notations used in this paper are similar to
those in Bena\"{i}m et. al. \cite{Benaim05},
Aubin et. al. \cite{Aubin} and Borkar \cite{BorkarBook}. We present a few for easy reference.
\paragraph{}
Let $H$ be an upper semi-continuous, set-valued map on $\mathbb{R}^d$, where
for any $x \in \mathbb{R}^d$, $H(x)$ is compact and
convex valued. Note that we say that $H$ is upper semi-continuous when $x_n \to x$, 
$y_n \to y$ and $y_n \in H(x_n)$ $\forall n$
implies $y \in H(x)$.
Consider the differential inclusion (DI)
\begin{equation} \label{di}
\dot{x} \ \in \ H(x).
\end{equation}
We say that $\textbf{x} \in \sum$ if $\textbf{x}$ 
is an absolutely continuous map that satisfies (\ref{di}).
The \textit{set-valued semiflow}
$\Phi$ associated with (\ref{di}) is defined on $[0, + \infty) \times \mathbb{R}^d$ as:
$\Phi_t(x) = \{\textbf{x}(t) \ | \ \textbf{x} \in \sum , \textbf{x}(0) = x \}$. Let
$\mathcal{T} \times M \subset [0, + \infty) \times \mathbb{R}^k$ and define
\begin{equation}\nonumber
 \Phi_\mathcal{T}(M) = \underset{t\in \mathcal{T},\ x \in M}{\bigcup} \Phi_t (x).
\end{equation}
\indent
$M \subseteq \mathbb{R}^d$ is \textit{invariant} if for every $x \in M$ there exists 
a complete trajectory in $M$, say $\textbf{x} \in \sum$  with $\textbf{x}(0) = x$.
\\ \indent
Let $x \in \mathbb{R}^d$ and $A \subseteq \mathbb{R}^d$, then
$d(x, A) : = \inf \{\lVert a- y \rVert \ | \ y \in A\}$. We define the $\delta$-\textit{open neighborhood}
of $A$ by $N^\delta (A) := \{x \ |\ d(x,A) < \delta \}$. The 
$\delta$-\textit{closed neighborhood} of $A$ 
is defined by $\overline{N^\delta} (A) := \{x \ |\ d(x,A) \le \delta \}$.
\\ \indent
Let $M \subseteq \mathbb{R}^d$, the $\omega-limit$ \textit{set} be given by
$
 \omega_{\Phi}(M) := \bigcap_{t \ge 0} \ \overline{\Phi_{[t, +\infty)}(M)}.
$
Similarly the \textit{limit set} of a solution $\textbf{x}$ is given by
$L(x) = \bigcap_{t \ge 0} \ \overline{\textbf{x}([t, +\infty))}$.
\\ \indent
$A \subseteq \mathbb{R}^d$ is an \textit{attractor} if it is compact, invariant
and there exists a neighborhood $U$ such that for any $\epsilon > 0$,
$\exists \ T(\epsilon) \ge 0$ such that $\Phi_{[T(\epsilon), +\infty)}(U) \subset
N^{\epsilon}(A)$. Such a $U$ is called the \textit{fundamental neighborhood} of $A$. The \textit{basin
of attraction } of $A$ is given by $B(A) = \{x \ | \ \omega_\Phi(x) \subset A\}$.
If $B(A) = \mathbb{R}^d$, then the set is called a \textit{globally attracting set}. It is called
\textit{Lyapunov stable} if for all $\delta > 0$, $\exists \ \epsilon > 0$ such that
$\Phi_{[0, +\infty)}(N^\epsilon(A)) \subseteq N^\delta(A)$.
\\ \indent
A set-valued map $h: \mathbb{R}^n \to \{subsets\ of\ \mathbb{R}^m$\} 
is called a \textit{Marchaud map} if it satisfies
the following properties:
\begin{itemize}
 \item[(i)] For each $z$ $\in \mathbb{R}^{n}$, $h(z)$ is convex and compact.
 \item[(ii)] \textit{(point-wise boundedness)} For each $z \in \mathbb{R}^{n}$,  
 $\underset{w \in h(z)}{\sup}$ $\lVert w \rVert$
 $< K \left( 1 + \lVert z \rVert \right)$ for some $K > 0$.
 \item[(iii)] $h$ is an \textit{upper semi-continuous} map. 
\end{itemize}
\indent The open ball of radius $r$ around $0$ is represented by $B_r(0)$,
while the closed ball is represented by $\overline{B}_r(0)$.
\subsection{Assumptions}\label{assumptions}
Recall that we have the following coupled recursion:
\begin{equation}\nonumber
\begin{split}
x_{n+1} = x_{n} + a(n) \left[ u_{n} + M^{1}_{n+1} \right], \\
y_{n+1} = y_{n} + b(n) \left[ v_{n} + M^{2}_{n+1} \right],
\end{split}
\end{equation}
where
$u_{n} \in h(x_{n}, y_{n})$, $v_{n} \in g(x_{n}, y_{n})$, $h: \mathbb{R}^{d+k} \to 
\left\{ subsets \ of\ \mathbb{R}^{d} \right\}$ and $g: \mathbb{R}^{d+k} \to \left\{ subsets\ of\ \mathbb{R}^{k} \right\}$.
\\ \indent
We list below our assumptions.
\begin{itemize}
 \item[(A1)] $h$ and $g$ are \textit{Marchaud maps}.
  \item[(A2)] $\{ a(n) \}_{n \ge 0}$ and $\{ b(n) \}_{n \ge 0}$ are two scalar sequences
  such that: \\ $a(n), b(n) > 0$, for all $n$, $\underset{n \ge 0}{\sum} \left(a(n) + b(n) \right) = \infty$,
 $\underset{n \ge 0}{\sum} \left( a(n)^{2} + b(n)^2 \right) < \infty$
 and $ \lim_{n \to \infty} \frac{b(n)}{a(n)} = 0$. 
 Without loss of generality, we let
 $\sup_n \ a(n),\ \sup_n \ b(n) \le 1$.
 \item[(A3)]$\{ M^{i}_{n}\}_{n \ge 1}$, $i = 1,2$, are square integrable martingale difference 
 sequences with respect to
 the filtration
 $\mathcal{F}_{n}$ $:=$ $ \sigma$ $\left( x_{m}, y_{m}, M^1 _{m}, M^2_{m}: \ m \le n \right)$, 
$n \ge 0$, such that
 $E[\lVert M^i _{n+1} \rVert ^{2} | \mathcal{F}_{n}]$ $\le$ $K 
 \left( 1 + \left( \lVert x_{n} \rVert  + \lVert y_{n} \rVert \right)^{2} \right)$, $i = 1, 2$, 
 for some constant $K > 0$. Without loss of generality assume that the same constant, $K$,
 works for both $(A1)$ (in the property $(ii)$ of Marchaud maps, see section~\ref{definition})
 and $(A3)$.
 \item[(A4)] $\sup_{n} \left\{ \lVert x_{n} \rVert + \lVert y_{n} \rVert \right\} < \infty$ $a.s.$
 \item[(A5)] For each $y \in \mathbb{R}^k$, the differential inclusion
 $\dot{x}(t) \in h(x(t), y)$ has a globally attracting set, $A_{y}$, that is also Lyapunov stable.
 Further, $\underset{x \in A_y}{\sup} \lVert x \rVert \le K \left(1 + \lVert y \rVert \right)$. The set-valued
 map $\lambda: \mathbb{R}^k \to \{subsets\ of\ \mathbb{R}^d \}$, where $\lambda(y) = A_y$, is
 upper semi-continuous.
\end{itemize}
 Define for each $y \in \mathbb{R}^{k}$, a function $G(y)\ := \overline{co}$
 $\left( \underset{x \in \lambda(y)}{\bigcup} g(x,y) \right)$.
The convex closure of a set $A \subseteq \mathbb{R}^k$, denoted by $\overline{co}(A)$, is
 closure of the convex hull of $A$, $i.e.,$ the 
closure of the smallest convex set containing $A$. It will be shown later that $G$ is a Marchaud map.

\begin{itemize}
 \item[(A6)] $\dot{y}(t) \in G(y(t))$ has a globally attracting set, $A_0$, that is also Lyapunov
 stable.
\end{itemize}
\indent
With respect to the faster timescale, the slower timescale iterates appear stationary, hence
the faster timescale iterates track a solution to $\dot{x}(t) \in h(x(t), y_0)$, where $y_0$ is fixed
(see Theorem~\ref{fastx}).
The $y$ iterates track a solution to $\dot{y}(t) \in G(y(t))$
(see Theorem~\ref{slowy}). 
It is worth noting that Theorems~\ref{fastx} \& \ref{slowy} only require $(A1)-(A5)$ to hold.
Since $G(\cdot)$ is the convex closure of a union of compact convex sets one can expect
the set-valued map to be point-wise bounded and convex. However, it is unclear why it should be
upper semi-continuous (hence Marchaud). In lemma~\ref{gismarchaud} we prove that
$G$ is indeed Marchaud without any additional assumptions.
\paragraph{}
Over the course of this paper we shall see that
$(A5)$ is the key assumption that links the asymptotic
behaviors of the faster and slower timescale iterates. 
It may be noted that $(A5)$ is 
\textit{weaker} than the corresponding assumption - $(B6)/(B6)'$ used in \cite{Perkins}.
For example, $(B6)'$ requires that $\lambda(y)$ and $\underset{x \in \lambda(y)}{\cup} g(x,y)$
be \textit{convex} for every $y \in \mathbb{R}^k$
while $(B6)$ requires that $\lambda(y)$ be singleton for every $y \in \mathbb{R}^k$. 
The reader is referred to
\cite{Perkins} for more details. 
Note that $\lambda(y)$ being a singleton is a strong requirement in itself 
since it is the
global attractor of some $DI$.
It is observed in most applications that both $\lambda(y)$ and $\underset{x \in \lambda(y)}{\cup} g(x,y)$
will not be convex and therefore $(B6)/(B6)'$ are easily violated. Further,
our application discussed in Section~\ref{Lagrange} illustrates the same.
\section{Proof of convergence}\label{convergenceproof}
Before we start analyzing the coupled recursion given by (\ref{twotimescale}),
we prove a bunch of auxiliary results.
\begin{lemma} \label{compactsasfn}
 Consider the differential inclusion $\dot{x}(t) \in H(x(t))$, where 
 $H: \mathbb{R}^n \to \mathbb{R}^n$ is a Marchaud map. Let $A$ be the associated globally attracting
 set that is also Lyapunov stable. Then $A$ is an attractor and every compact set
 containing $A$ is a fundamental neighborhood.
\end{lemma}
\begin{proof}
Since $A$ is compact and invariant, it is left to prove the following:
 given a compact set $K \subseteq \mathbb{R}^n$ such that $A \subseteq K$; for each
 $\epsilon > 0$ there exists $T(\epsilon) > 0$  such that $\Phi_{t}(K) \subseteq N^\epsilon (A)$ for all $t \ge T(\epsilon)$.
\\ \indent
 Since $A$ is Lyapunov stable, corresponding to $N^\epsilon (A)$ there exists $N^\delta (A)$,
 where $\delta > 0$, such that $\Phi_{[0, + \infty)} (N^\delta (A)) \subseteq N^\epsilon (A)$.
 Fix $x_{0} \in K$. Since $A$ is a globally attracting set, $\exists t(x_0) > 0$ such that 
 $\Phi_{t(x_0)}(x_0) \subseteq N^{\delta / 4} (A)$. Further, 
 from the upper semi-continuity of flow it follows that 
 $\Phi_{t(x_0)}(x) \subseteq N^{\delta / 4} (\Phi_{t(x_0)}(x_0))$ for all $x \in N^{\delta(x_0)}(x_0)$,
 where $\delta(x_0) > 0$, see \textit{Chapter 2} of Aubin and Cellina \cite{Aubin}. 
 Hence we get $\Phi_{t(x_0)}(x) \subseteq N^\delta (A)$. Further since $A$ is Lyapunov stable, we get
 $\Phi_{(t(x_0), + \infty]} (x) \subseteq N^\epsilon (A)$. In this manner for each $x \in K$
 we calculate $t(x)$ and $\delta(x)$, the collection
 $\left\{N^{\delta(x)}(x) : x \in K \right\}$
 is an open cover for $K$. Since $K$ is compact, there exists a finite sub-cover
 $\left\{N^{\delta(x_{i})}(x_{i}) \ |\  1 \le i \le m \right\}$. For 
 $T(\epsilon) := max \{ t(x_{i}) \ |\  1 \le i \le m \}$, we have
 $\Phi_{[T(\epsilon), + \infty)} (K) \subseteq N^\epsilon(A)$.
 
\end{proof}
In Theorem~\ref{slowy} we prove that the slower timescale trajectory
asymptotically tracks a solution to $\dot{y}(t) \in G(y(t))$. The following lemma
ensures that the aforementioned $DI$ has at least one solution.

\begin{lemma} \label{gismarchaud}
 The map $G$ referred to in (A6) is a Marchaud map.
\end{lemma}
\begin{proof}
Fix an arbitrary $y \in \mathbb{R}^k$.
 For any $x \in \lambda(y)$, it follows from $(A1)$ that
 \begin{equation}\nonumber
  \underset{z \in g(x,y)}{\sup} \lVert z \rVert
 \le K(1 + \lVert x\rVert + \lVert y \rVert ). 
 \end{equation}

 From assumption $(A5)$, we have that $\lVert x \rVert
 \le K(1 + \lVert y \rVert)$. Substituting in the above equation we may conclude the following:
 
 \begin{equation} \nonumber
\begin{split}
 & \underset{ z \in g(x,y)}{\sup} \lVert z \rVert \le K \left(1 + K \left(1 + \lVert y\rVert \right) + \lVert y \rVert \right)
 = K(K+1) (1 + \lVert y \rVert) \ , \\
& \underset{ z \in \underset {x \in \lambda(y)}{\bigcup} g(x,y)}{\sup} \lVert z \rVert \le 
 K(K+1) (1 + \lVert y \rVert) \ , \\
& \underset{ z \in G(y)}{\sup} \lVert z \rVert \le 
 K(K+1) (1 + \lVert y \rVert).
\end{split}
\end{equation}
We have thus proven that $G$ is point-wise bounded.
From the definition of $G$, it follows that $G(y)$ is convex and compact.
\\ \indent
It remains to show that $G$ is an upper semi-continuous map.
Let $z_{n} \to z$ and $y_{n} \to y$ in $\mathbb{R}^k$ with $z_{n} \in G(y _n)$, 
 $\forall$ $n \ge 1$.
 We need to show that $z \in G(y)$. We present a proof by contradiction.
 Since $G(y)$ is convex and compact, $z \notin G(y)$
 implies that there exists a linear functional on $\mathbb{R}^{k}$, say $f$, such that
 $\underset{w \in G(y)}{\sup}$ $f(w) \le \alpha - \epsilon$
 and $f(z) \ge \alpha + \epsilon$, for some
 $\alpha \in \mathbb{R}$ and $\epsilon > 0$. Since $z_{n} \to z$, there exists 
 $N$ such that for all $n \ge N$, $f(z_{n}) \ge \alpha + \frac{\epsilon}{2}$. In other
 words, $G(y_{n}) \cap  [f \ge \alpha + \frac{\epsilon}{2}] \neq \phi$ for
 all $n \ge N$. Here the notation $[f \ge a]$ is used to denote the set
 $\left\{ x \ |\ f(x) \ge a \right\}$.
 \\ \indent
 For the sake of convenience, we denote the set
 $\underset{x \in \lambda(y)}{\bigcup} g(x,y)$ by $B(y)$.
 We claim that $B(y_{n}) \cap [f \ge \alpha + \frac{\epsilon}{2}] \neq \phi$
 for all $n \ge N$. We prove this claim later,
 for now we assume that the claim
 is true and proceed.
Pick $w_{n} \in g(x_{n}, y_{n}) \cap [f \ge \alpha + \frac{\epsilon}{2}]$, where $x_{n} \in \lambda(y_{n})$
 and $n \ge N$ . 
 It can be shown that $\{x_{n}\}_{n \ge N}$ and $\{w_{n}\}_{n \ge N}$ are norm bounded sequences
 and hence contain convergent sub-sequences. Construct sub-sequences,
 $\{w_{n(k)}\}_{k \ge 1} \subseteq \{w_{n}\}_{n \ge N}$ and
  $\{x_{n(k)}\}_{k \ge 1} \subseteq \{x_{n}\}_{n \ge N}$ such that
 $\underset{k \to \infty}{\lim} w_{n(k)} = w$ and $\underset{k \to \infty}{\lim} x_{n(k)} = x$.
 It follows from the upper semi-continuity of $g$ that $w \in g(x,y)$ and from the upper
 semi-continuity of $\lambda$ that $x \in \lambda(y)$, hence
 $w \in G(y)$. Since $f$ is continuous, $f(w) \ge \alpha + \frac{\epsilon}{2}$. This is a contradiction.
 \\ \indent
 It remains to prove that  $B(y_{n}) \cap [f \ge \alpha + \frac{\epsilon}{2}] \neq \phi$
 for all $n \ge N$. Suppose this were false, then
 $\exists \{m(k)\}_{k \ge 1} \subseteq \{n \ge N\}$ 
 such that $B(y_{m(k)}) \subseteq [f < \alpha + \frac{\epsilon}{2}]$
 for each $k \ge 1$. It can
 be shown that $\overline{co}(B(y_{m(k)})) \subseteq 
 [f \le \alpha + \frac{\epsilon}{2}]$ for each $k \ge 1$. 
 Since $z_{m(k)} \to z$, $\exists N_{1}$ such that for all $m(k) \ge N_1$, 
 $f(z_{m(k)}) \ge \alpha + \frac{3 \epsilon}{4}$. This is a contradiction.
 Hence we get $B(x_{n}) \cap [f \ge \alpha + \frac{\epsilon}{2}] \neq \phi$
 for all $n \ge N$.
 \end{proof}
 It is worth noting that $(A5)$ is a key requirement in the above proof. 
 In the next lemma, we show the convergence of the martingale noise terms.
 \begin{lemma} \label{noiseconv}
  The sequences $\{ \zeta^1 _n \}$ and $\{ \zeta^2 _n \}$, where $ \zeta^1 _n 
  =  \sum_{m=0}^{n-1} a(m) M^1_{m+1}$ and $ \zeta^2 _n 
  =  \sum_{m=0}^{n-1} b(m) M^2_{m+1}$, are convergent almost surely.
 \end{lemma}
 \begin{proof}
 Although a proof of the above statement can be found in \cite{Benaim96} or \cite{BorkarBook},
 we provide one for the sake of completeness.
  We only prove the almost sure convergence of $\zeta^1_n$ as the convergence of $\zeta^2_n$ 
  can be similarly shown. 
  \\ \indent
  It is enough to show that
  \begin{equation}\nonumber
  \begin{split}
   &\sum_{m=0}^{\infty} a(m)^2 E \left[ \lVert \zeta^1_{m+1} - \zeta^1 _{m} \rVert ^2 | \mathcal{F}_m \right]  \ < \ \infty \ a.s., 
   \\
   i.e., &\sum_{m=0}^{\infty} a(m)^2 E \left[ \lVert M^1_{m+1} \rVert ^2 | \mathcal{F}_m \right]  \ < \ \infty \ a.s.
   \end{split}
  \end{equation}
  From assumption $(A3)$ it follows that 
  \begin{equation} \nonumber
   \sum_{m=0}^{\infty} a(m)^2 E \left[ \lVert M^1_{m+1} \rVert ^2 | \mathcal{F}_m \right]  \ \le \
   K \sum_{m=0}^{\infty} a(m)^2 \left( 1 + (\lVert x_m\rVert + \lVert y_m\rVert)^2 \right).
  \end{equation}
From assumptions $(A2)$ and $(A4)$ it follows that 
\begin{equation} \nonumber
K\sum_{m=0}^{\infty} a(m)^2 \left( 1 + (\lVert x_m\rVert + \lVert y_m\rVert)^2 \right) < \infty \ 
a.s.
\end{equation}
 \end{proof}
We now prove a couple of technical results that are essential to the proofs
of Theorems ~\ref{fastx} and ~\ref{slowy}.
 \begin{lemma} \label{xclosetoay}
 Given any $y_0 \in \mathbb{R}^k$ and $\epsilon > 0$, there exists $\delta > 0$ such that
 for all $x \in N^\delta(\lambda(y_0))$, we have $g(x, y_0) \subseteq N^\epsilon (G(y_0))$.
 \end{lemma}
 \begin{proof}
 Assume the statement is not true. Then,
  $\exists$ $\delta_n \downarrow 0$ and $x_{n} \in N^{\delta_n}(\lambda(y_0))$ such that
  $g(x_n, y_0) \nsubseteq N^\epsilon (G(y_0 ))$, $n \ge 1$. 
  In other words, $\exists \gamma_n \in g(x_n, y_0)$ and $\gamma_n \notin N^\epsilon (G(y_0))$ for each $n \ge 1$.
  Since $\{x_n\}$ and $\{ \gamma_n \}$ are bounded sequences there exist convergent sub-sequences,
  $\underset{k \to \infty} {\lim} x_{n(k)} = x$ and
  $\underset{k \to \infty} {\lim} \gamma_{n(k)} = \gamma$. 
  Since $x_{n(k)} \in N^{\delta_{n(k)}} (\lambda(y_0))$ and $\delta_{n(k)} \downarrow 0$ it
  follows that $x \in \lambda(y_0)$ and hence $g(x, y_0) \subseteq G(y_0)$. We also have that
  $v \notin N^\epsilon (G(y_0))$ as $v_{n(k)} \notin N^\epsilon (G(y_0))$ for all $k \ge 1$.
  Since $g$ is
  upper semi-continuous it follows that $\gamma \in g(x, y_0)$ and hence 
  $\gamma \in G(y_0)$. This is a contradiction.
 \end{proof}

\begin{lemma} \label{seqto}
 Let $x_0 \in \mathbb{R}^d$ and $y_0 \in \mathbb{R}^k$ be such that the statement of
 lemma ~\ref{xclosetoay} is satisfied (with $x_0$ in place of $x$). If $\underset{n \to \infty}{\lim} x_n = x_0$ and
 $\underset{n \to \infty}{\lim} y_n = y_0$ then $\exists N$ such that $\forall n \ge N$,
 $g(x_n, y_n) \subseteq N^\epsilon (G(y_0))$.
\end{lemma}
\begin{proof}
 If not, $\exists$ $\{n(k)\} \subseteq \{n\}$ such that $\underset{k \to \infty}{\lim} n(k) = \infty $
 and $g(x_{n(k)}, y_{n(k)}) \nsubseteq N^\epsilon (G(y_0))$.
 Without loss of generality assume that $\{n(k)\} = \{n\}$. In other words,
 $\exists \gamma_n \in g(x_n, y_n)$ such that $\gamma_n \notin N^\epsilon (G(y_0))$ for all $n \ge 1$.
 Since $\{\gamma_n\}$ is a bounded sequence, it has a convergent sub-sequence, $i.e.$, 
 $\underset{m \to \infty}{\lim} \gamma_{n(m)} = \gamma$. 
 Since $\underset{m \to \infty}{\lim} x_{n(m)} = x_0$, $\underset{m \to \infty}{\lim} y_{n(m)} = y_0$ and
 $g$ is upper semi-continuous it
 follows that $\gamma \in g(x_0, y_0)$ and finally from lemma \ref{xclosetoay} 
 we get that $\gamma \in N^\epsilon (G(y_0))$.
 This is a contradiction.
\end{proof}
Before we proceed let us construct trajectories, using (\ref{twotimescale}), with respect to the faster timescale.
Define $t(0) := 0$, $t(n) \ := \ \sum_{i=0}^{n-1} a(i)$, $n \ge 1$.
The linearly interpolated trajectory $\overline{x}(t)$, $t \ge 0$,
is constructed from  the sequence $\{ x_{n} \}$ as follows: let $\overline{x}(t(n)) \ := \ x_{n}$ and for 
$t \ \in \ (t(n), t(n+1))$, let 
\begin{equation}\label{xbar}
\overline{x}(t) \ := \ 
\left( \frac{t(n+1) - t}{t(n+1) - t(n)} \right)\ \overline{x}(t(n))\ +\ 
\left( \frac{t - t(n)}{t(n+1) - t(n)} \right) \  
\overline{x}(t(n+1)). 
\end{equation}
We construct a piecewise constant trajectory from the sequence $\{u_{n}\}$
as follows: $\overline{u}(t) := u_{n}$ for $t \in [t(n), t(n+1))$, $n \ge 0$.
\paragraph{}
Let us construct trajectories with respect to the slower timescale in a similar manner.
Define $s(0) := 0$, $s(n) \ := \ \sum_{i=0}^{n-1} b(i)$, $n \ge 1$.
Let $\widetilde{y}(s(n)) \ := \ y_{n}$ and for 
$s \ \in \ (s(n), s(n+1))$, let 
\begin{equation}\label{ytilde}
\widetilde{y}(s) \ := \ 
\left( \frac{s(n+1) - s}{s(n+1) - s(n)} \right)\ \widetilde{y}(s(n))\ +\ 
\left( \frac{s - s(n)}{s(n+1) - s(n)} \right) \  
\widetilde{y}(s(n+1)). 
\end{equation}
Also $\widetilde{v}(s) := v_{n}$ for $s \in [s(n), s(n+1))$, $n \ge 0$, is the corresponding
piecewise constant trajectory.
\paragraph{}
For $s \ge 0$, let $x^s (t)$, $t \ge 0$, denote the solution to $\dot{x}^s (t) = \overline{u}
(s + t)$ with the initial condition $x^s (0) = \overline{x}(s)$. Similarly,
let $y^s (t)$, $t \ge 0$, denote the solution to $\dot{y}^s (t) = \widetilde{v}
(s + t)$ with the initial condition $y^s (0) = \widetilde{y}(s)$.
\paragraph{}
The $y$ iterate in recursion (\ref{twotimescale}) can be re-written as
\begin{equation}\label{rewritey}
 y_{n+1} = y_n + a(n) \left[ \frac{b(n)}{a(n)} v_n + \frac{b(n)}{a(n)} M^{2}_{n+1} \right].
\end{equation}
Define $\epsilon(n) := \frac{b(n)}{a(n)} v_n$ and $M^3_{n+1} = \frac{b(n)}{a(n)} M^2_{n+1}$. 
It can be shown that the stochastic iteration given by $y_{n+1} = y_n + a(n) M^3_{n+1}$ satisfies 
the set of assumptions given in Bena\"{i}m \cite{Benaim96}. 
From $(A1)$, $(A2)$ and $(A4)$ it follows that $\epsilon(n) \to 0$ almost surely. Since
$\epsilon(n) \to 0$ the
recursion given by (\ref{rewritey}) and $y_{n+1} = y_n + a(n)M^3_{n+1}$ have the same
asymptotics. For a precise statement and proof the reader is referred to
lemma $2.1$ of \cite{Borkartt}. 
\\ \indent
Define $\overline{y}(t(n)) := y_{n}$, where $n \ge 0$ and
$\overline{y}(t)$ for
$t \ \in \ (t(n), t(n+1))$ by 
\begin{equation}\label{ybar}
\overline{y}(t) \ := \ 
\left( \frac{t(n+1) - t}{t(n+1) - t(n)} \right)\ \overline{y}(t(n))\ +\ 
\left( \frac{t - t(n)}{t(n+1) - t(n)} \right) \  
\overline{y}(t(n+1)). 
\end{equation}
The trajectory $\overline{y}(\cdotp)$ can be seen as an evolution of the $y$ iterate
with respect to the faster timescale, $\{a(n)\}$.

\begin{lemma} \label{fasty}
 Almost surely every limit point, $y(\cdotp)$, of $\{ \overline{y}(s+ \cdotp) \ |\  s \ge 0\}$
 in $C([0, \infty), \mathbb{R}^k)$ as $s \to \infty$ satisfies $y(t) = y(0),\ t \ge 0$.
\end{lemma}
\begin{proof}
 It can be shown that $y_{n+1} = y_n + a(n)M^3_{n+1}$ satisfies the assumptions of Bena\"{i}m 
 \cite{Benaim96}. Hence the corresponding linearly interpolated trajectory tracks the solution
 to $\dot{y}(t) = 0$.
 The statement of the lemma then follows trivially.
\end{proof}

\begin{lemma} \label{equicont}
 For any $T > 0$, $\underset{s \to \infty}{\lim} 
 \underset{t \in [0, T]}{\sup} \lVert \overline{x}(s+t) - x^s(t) \rVert = 0$
 and $\underset{s \to \infty}{\lim}
  \underset{t \in [0, T]}{\sup} \lVert \widetilde{y}(s+t) - y^s(t) \rVert = 0$, $a.s.$
\end{lemma}
\begin{proof}
 In order to prove the above lemma, it enough to prove the following:
 \begin{equation} \nonumber
  \begin{split}
  & \underset{t(n) \to \infty}{\lim} \ 
 \underset{0 \le t(n+m) - t(n) \le T}{\sup} \lVert \overline{x}(t(n+m)) - x^{t(n)}(t(n+m) - t(n)) \rVert = 0 \mbox{ and} \\
  &\underset{s(n) \to \infty}{\lim} \ 
 \underset{0 \le s(n+m) - s(n) \le T}{\sup} \lVert \widetilde{y}(s(n+m)) - y^{s(n)}(s(n+m) - s(n)) \rVert = 0 \ a.s.
  \end{split}
 \end{equation}
 
Note the following:
\begin{equation}\nonumber
 \begin{split}
  &\overline{x}(t(n+m))  = \overline{x}(t(n)) + \sum_{k = 0}^{m-1} \left[ a(n+k) \left( 
\overline{u}(t(n+k)) + M^1_{n+k+1} \right) \right],  \\
  &x^{t(n)}(t(n+m) - t(n))  = \overline{x}(t(n)) + \int_{0}^{t(n+m) - t(n)} \overline{u}(t(n) + z)\  dz ,
 \end{split}
\end{equation}
\begin{equation}\label{equicont1}
 x^{t(n)}(t(n+m) - t(n)) = \overline{x}(t(n)) + \int_{t(n)}^{t(n+m)} \overline{u}(z)\  dz.
\end{equation}
From (\ref{equicont1}), we get,
\begin{multline} \nonumber
\lVert \overline{x}(t(n+m)) - x^{t(n)}
\left( t(n+m) - t(n) \right) \rVert = \\
  \left\lVert \sum_{k = 0}^{m-1} a(n+k) \overline{u}(t(n+k)) 
  - \sum_{k = 0}^{m-1} \int_{t(n+k)}^{t(n+k+1)} \overline{u}(z)\  dz
+ \sum_{k = 0}^{m-1} a(n+k) M^1_{n+k+1}  \right\rVert.
\end{multline}
The $R.H.S.$ of the above equation equals $\left\lVert \sum_{k = 0}^{m-1} a(n+k) M^1_{n+k+1}  \right\rVert$
as \[\sum_{k = 0}^{m-1} a(n+k) \overline{u}(t(n+k)) \ = 
\sum_{k = 0}^{m-1} \int_{t(n+k)}^{t(n+k+1)} \overline{u}(z)\  dz.\]
Since $ \zeta^1 _n 
  :=  \sum_{m=0}^{n-1} a(m) M^1_{m+1}$, $n \ge 1$,
converges $a.s.$, the first part of claim follows.
\\ \indent
The second part, for the $y$ iterates, can be similarly proven.
\end{proof}

From assumptions $(A1)$ and $(A4)$ it follows that  
$\{ x^r(\cdotp) \ |\  r \ge 0\}$ and $\{ y^r(\cdotp) \ |\  r \ge 0\}$ are equicontinuous 
and pointwise bounded
families of functions. By the Arzela-Ascoli theorem they are relatively compact in $C([0, \infty), 
\mathbb{R}^d)$ and $C([0, \infty), \mathbb{R}^k)$ respectively. From lemma ~\ref{equicont} it
then follows that $\{\overline{x}(r + \cdotp) \ |\  r \ge 0 \}$ and 
$\{\widetilde{y}(r + \cdotp) \ |\  r \ge 0 \}$ are also relatively compact, 
see (\ref{xbar}) and (\ref{ytilde}) for the 
definitions of $\overline{x}(\cdotp)$ and $\widetilde{y}(\cdotp)$, respectively.

\subsection{Convergence in the faster timescale}
The following theorem
and its proof are similar to
Theorem $2$ from Chapter $5$ of Borkar \cite{BorkarBook}. 
We present a proof for the sake of completeness. 

\begin{theorem} \label{fastx}
 Almost surely, every limit point of $\{ \overline{x}(r+ \cdotp) \ |\  r \ge 0\}$ in $C([0, \infty), \mathbb{R}^d)$
 is of the form $x(t) = x(0) + \int_0^t u(z) \ dz$, where $u$ is a measurable function
 such that $u(t) \in h(x(t), y(0))$, $t \ge 0$, for some fixed $y(0) \in \mathbb{R}^k$.
\end{theorem}
\begin{proof}
 Fix $T > 0$, then $\{ \overline{u}(r + t) \ |\  t \in [0, T] \}$, $r \ge 0$ can be viewed
 as a subset of $L_2 ([0,T], \mathbb{R}^d)$. From $(A1)$ and $(A4)$ it follows that
 the above is uniformly bounded and hence weakly relatively compact. Let $\{r(n)\}$ be a sequence such that
 the following hold:
 \begin{itemize}
  \item[(i)] $\underset{n \to \infty}{\lim} r(n) = \infty$.
  \item[(ii)] There exists some $x(\cdotp) \in C([0, \infty), \mathbb{R}^d)$ 
  such that $\overline{x}(r(n) + \cdotp) \to x(\cdotp)$
 in $C([0, \infty), \mathbb{R}^d)$. This is because $\{\overline{x}(r + \cdotp) \ | \ 
 r \ge 0\}$ is relatively compact in $C([0, \infty), \mathbb{R}^d)$.
  \item[(iii)]$\overline{y}(r(n) + \cdotp) \to y(\cdotp)$ in $C([0, \infty), \mathbb{R}^k)$
  for some $y \in C([0, \infty), \mathbb{R}^k)$. It follows from lemma~\ref{fasty}
  that $y(t) = y(0)$ for all $t \ge 0$.
  \item[(iv)] $\overline{u}(r(n) + \cdotp) \to u(\cdotp)$ weakly in $L_2 ([0,T], \mathbb{R}^d)$.
 \end{itemize}
 From lemma ~\ref{equicont}, it follows that $x^{r(n)}(\cdotp) \to x(\cdotp)$
 in $C([0, \infty), \mathbb{R}^d)$, and we have that
 $\int_0 ^t \overline{u}(r(n) + z) \ dz \ \to
 \int_0 ^t u(z) \ dz$ for $t \in [0,T]$. Letting $n \to \infty$ in
 \begin{equation}\nonumber
  x^{r(n)}(t) = x^{r(n)}(0) + \int_0 ^t \overline{u}(r(n) + z) \ dz ,\ t \in [0,T],
 \end{equation}
we get $x(t) = x(0) + \int_0 ^t u(z) \ dz, \ t \in [0,T]$. 
\\ \indent
Since $\overline{u}(r(n)+ \cdotp )
\to u(\cdotp)$ weakly in $L_2 ([0,T], \mathbb{R}^d)$, there exists
$\{n(k)\} \subset \{n\}$ such that $n(k) \uparrow \infty$ and
\begin{equation}\nonumber
 \frac{1}{N} \sum_{k=1}^{N} \overline{u}(r(n(k)) + \cdotp) \to u(\cdotp)
\end{equation}
strongly in $L_2 ([0,T], \mathbb{R}^d)$. Further, there exist $\{N(m)\} \subset \{N\}$
such that $N(m) \uparrow \infty$ and
\begin{equation} \label{aeconv}
 \frac{1}{N(m)} \sum_{k=1}^{N(m)} \overline{u}(r(n(k)) + \cdotp) \to u(\cdotp)
\end{equation}
$a.e.$ in $[0,T]$. 
\\ \indent
Define $[t] := max\{t(n) \ |\  t(n) \le t\}$. If we fix $t_0 \in [0, T]$
such that (\ref{aeconv}) holds, then $\overline{u}(r(n(k)) + t_0) \in 
h(\overline{x}([r(n(k)) + t_0]), \overline{y}([r(n(k)) + t_0]))$ for  $k \ge 1$. Since
$\underset{n(k) \to \infty}{\lim}
\lVert \overline{x}(r(n(k)) + t_0) - \overline{x}([r(n(k)) + t_0]) \rVert = 0$, it follows
that $\underset{k \to \infty}{\lim} \overline{x}([r(n(k)) + t_0]) = x(t_0)$, and similarly, we have that
$\underset{k \to \infty}{\lim} \overline{y}([r(n(k)) + t_0]) = y(0)$.
Since $h$ is upper semi-continuous it follows that $\underset{k \to \infty}{\lim}$
$d\left(\overline{u}(r(n(k)) + t_0) , h(x(t_0), y(0))\right) = 0$. The set $h(x(t_0), y(0))$ is compact and
convex, hence it follows from (\ref{aeconv}) that $u(t_0) \in h(x(t_0), y(0))$.
\end{proof}

\subsection{Convergence in the slower timescale}
\begin{theorem} \label{slowy}
 For any  $\epsilon > 0$, almost surely any limit point of $\{\widetilde{y}(r + \cdotp) \ |\  r \ge 0\}$
 in $C([0, \infty), \mathbb{R}^k)$ is of the form $y(t) = y(0) + \int _0 ^t v(z) \ dz$, where
 $v$ is a measurable function such that $v(t) \in N^{\epsilon}(G(y(t)))$, $t \ge 0$.
\end{theorem}
\begin{proof}
Fix $T > 0$. As before let
$\{r(n)\}_{n \ge 1}$ be a sequence such that the following hold:
\begin{itemize}
 \item[(i)] $\underset{n \to \infty}{\lim} r(n) = \infty$.
 \item[(ii)] $\widetilde{y}(r(n) + \cdotp) \to y(\cdotp)$
 in $C([0, \infty), \mathbb{R}^k)$, where $y(\cdotp) \in C([0, \infty), \mathbb{R}^k)$.
 \item[(iii)] $\widetilde{v}(r(n) + \cdotp) \to v(\cdotp)$ weakly in $L_2 ([0,T], \mathbb{R}^k)$.
\end{itemize}
Also, as before,
 we have the following:
 \begin{itemize}
  \item[(i)] There exists $\{n(k)\} \subseteq \{n\}$ such that
 $\frac{1}{N} \sum_{k=1}^{N} \widetilde{v}(r(n(k)) + \cdotp) \to v(\cdotp)
$
strongly in $L_2 ([0,T], \mathbb{R}^d)$ as $N \to \infty$.
\item[(ii)] There exist $\{N(m)\} \subset \{N\}$
such that $N(m) \uparrow \infty$ and
\begin{equation} \label{aeconv1}
 \frac{1}{N(m)} \sum_{k=1}^{N(m)} \widetilde{v}(r(n(k)) + \cdotp) \to v(\cdotp)
\end{equation}
 $a.e.$ on $[0,T]$.
 \end{itemize}
 Define $[s]' := max\{s(n) \ |\  s(n) \le s\}$. Construct a sequence $\{m(n)\}_{n \ge 1}$ 
 $\subseteq \mathbb{N}$ such that $s(m(n)) = [r(n) + t_0]'$ for each $n \ge 1$.
Observe that $\overline{y}(t(m(n)))
 = \widetilde{y}(s(m(n)))$ and $\widetilde{v}(r(n) + t_0) \in 
 g(\overline{x}(t(m(n))), \overline{y}(t(m(n))))$. 
 \\ \indent
 Choose $t_0 \in (0, T)$ such that (\ref{aeconv1}) is satisfied.
 If we show that $\exists \ N$ such that
 for all $n \ge N$, $g(\overline{x}(t(m(n))), \overline{y}(t(m(n)))) \subseteq N^\epsilon 
 (G(y(t_0)))$ then (\ref{aeconv1}) implies that $v(t_0) \in \overline{N^\epsilon} (G(y(t_0)))$.
 \\ \indent
 It remains to show the existence of such a $N$.
 We present a proof by contradiction. 
 We may assume without loss of generality that
 for each $n \ge 1$, $g(\overline{x}(t(m(n))), \overline{y}(t(m(n)))) \nsubseteq N^\epsilon 
 (G(y(t_0)))$, $i.e.$,
 $\exists$ $\gamma_n \in g(\overline{x}(t(m(n))), \overline{y}(t(m(n))))$
 such that $\gamma_n \notin N^\epsilon (G(y(t_0)))$.
 Let $S_1$ be the set on which $(A4)$ is satisfied and $S_2$ be the set on which 
 lemma~\ref{noiseconv} holds. Clearly $P(S_1 \cap S_2) = 1$. For each $\omega \in S_1 \cap S_2$,
 $\exists$ $R(\omega) < \infty$ such that $\sup_n \lVert x_n (\omega) + y_n (\omega)  \rVert 
 \le R(\omega)$ and $\sup_n \ K(1 + \lVert y_n (\omega) \rVert)  \le R(\omega)$. 
 In what follows we merely use $R$ and the dependence on $\omega$ (sample path)
 is understood to be implicit.
 From lemma ~\ref{compactsasfn} it
 follows that corresponding to $\dot{x}(t) \in h(x(t), y(t_0))$ and some $\delta > 0$
 there exists $T_0$, possibly dependent on $R$, such that
 for all $t \ge T_0$, $\Phi_t (x_0) \in N^\delta (\lambda(y(t_0)))$ for all $x_0$ 
 $\in \overline{B}_R(0)$.
 \\ \indent
 We construct a new sequence $\{l(n)\}_{n \ge 1}$ from $\{m(n)\}_{n \ge 1}$ such that
 $t(l(n)) = min\{ t(m) \ |\  |t(m(n)) - t(m)| \le T_0\}$. Since
$\{\overline{x}(r + \cdotp) \ |\  r \ge 0 \}$ is relatively compact in $C([0, \infty), \mathbb{R}^d)$,
  it follows that
 $\overline{x}(t(l(n)) + \cdotp) \to x(\cdotp)$ in $C([0,T_0], \mathbb{R}^d)$.
 From lemma~\ref{fasty} we can conclude that 
 $\overline{y}(t(l(n)) + \cdotp) \to y( \cdotp)$ in $C([0,T_0], \mathbb{R}^k)$, where
 $y(t) = y(t_0)$ for all $t \in [0, T_0]$. lemma~\ref{fasty} only asserts that
 the limiting function is a constant, we recognize this constant to be $y(t_0)$ since
$\lVert \overline{y}(t(l(n)) + T_0) - \overline{y}(t(m(n))) \rVert \to 0$ 
 and $\overline{y}(t(l(n)) + T_0) \to y(t_0)$.
 Note that in the foregoing discussion 
 we can only assert the existence of convergent subsequences, again for the sake of
 convenience we assume that the sequences at hand are both convergent.
 It follows from Theorem~\ref{fastx} that $x(t) = x(0) + \int_0 ^t u(z) \ dz$,
 where $u(t) \in h(x(t), y(t_0))$. Since $x(0) \in \overline{B}_R(0)$ 
 it follows that $x(T_0) \in N^\delta (\lambda(y(t_0)))$. 
 \\ \indent
 From lemma ~\ref{xclosetoay}
 we get
 $g(x(T_0), y(t_0)) \subseteq N^\epsilon(G(y(t_0)))$.
 Since $\lVert \overline{x}(t(m(n))) - \overline{x}(t(l(n)) + T_0) \rVert \to 0$
 it  follows that $\overline{x}(t(m(n))) \to x(T_0)$. It follows from lemma ~\ref{seqto}
 that $\exists N$ such that for $n \ge N$, $g(\overline{x}(t(m(n))), \overline{y}(t(m(n))))
 \subseteq N^\epsilon (G(y(t_0)))$. This is a contradiction.
\end{proof}
A direct consequence of  the above theorem is that
almost surely any limit point of $\{\widetilde{y}(r + \cdotp) \ |\  r \ge 0\}$
 in $C([0, \infty), \mathbb{R}^k)$ is of the form $y(t) = y(0) + \int _0 ^t v(z) \ dz$, where
 $v$ is a measurable function such that $v(t) \in G(y(t))$, $t \ge 0$.

\subsection{Main result}\label{mainresult}

\begin{theorem}\label{main}
 Under assumptions $(A1)-(A6)$, almost surely the set of accumulation points is given by
 \begin{equation} \label{maineq}
   \left\{(x,y) \ | \ \underset{n \to \infty}{\underline{\lim}} \ d\left( (x,y), (x_n, y_n) \right)
 =0 \right\}
 \subseteq \underset{y \in A_0}{\bigcup} \left\{(x,y) \ | \ x \in \lambda(y) \right\}.
 \end{equation}
\end{theorem}
\begin{proof}
 The statement follows directly from Theorems ~\ref{fastx} and ~\ref{slowy}.
\end{proof}
Note that assumption $(A6)$ allows us to narrow the set of interest. If $(A6)$ does not hold then
we can only conclude that the $R.H.S.$ of (\ref{maineq}) is
 $\underset{y \in \mathbb{R}^k}{\bigcup} \left\{(x,y) \ | \ x \in \lambda(y) \right\}$.
 On the other hand if $(A6)$ holds and $A_0$ consists of a single point, say $y_0$, then the 
 $R.H.S.$ of (\ref{maineq}) is $\{(x,y_0) \ | \ x \in \lambda(y_0) \}$. Further, if 
 $\lambda(y_0)$ is of cardinality one then the $R.H.S.$ of (\ref{maineq}) is just $(\lambda(y_0), y_0)$.
 \paragraph{}
\textit{Remark:}It may be noted that all proofs and conclusions in this paper will go through if 
$(A1)$ is weakened to let  $g$ be upper semi-continuous and $g(x, \cdotp)$ be Marchaud on $\mathbb{R}^k$ for each fixed $x \in \mathbb{R}^d$.
 
 \section{Application: An SA algorithm to solve the Lagrangian dual problem} \label{Lagrange}
 Let $f: \mathbb{R}^d \to \mathbb{R}$ and $g: \mathbb{R}^d \to \mathbb{R}^k$ be two given
 functions. We want to minimize $f(x)$ subject to the condition that $g(x) \le 0$ (every
 component of $g(x)$ is non-positive). This problem can be stated in the following primal form:
 \begin{equation}\label{primal}
  \underset{x \in \mathbb{R}^d}{inf} \ \underset{\substack{\mu \in \mathbb{R}^k \\ \mu \ge 0}}{sup}
  \left( f(x) + \mu ^{\mathsmaller T} g(x) \right).
 \end{equation}
Let us consider the following two timescale $SA$ algorithm to solve the primal (\ref{primal}):
\begin{equation}\label{primalSA}
\begin{split}
\mu_{n+1} = \mu_{n} + a(n) \left[ \nabla_{\mu}  \left( f(x_n) + \mu_n ^{\mathsmaller T} g(x_n) \right) + M^1 _{n+1} \right], \\
 x_{n+1} = x_n - b(n) \left[ \nabla_{x} \left( f(x_n) + \mu_n ^{\mathsmaller T} g(x_n) \right) + M^2 _{n+1} \right].
 \end{split}
\end{equation}
where,
$a(n), b(n) > 0$, 
$\sum_{n \ge 0} a(n) = \sum_{n \ge 0} b(n) = \infty$, $\sum_{n \ge 0} a(n) ^2 < \infty$,
$\sum_{n \ge 0} b(n) ^2 < \infty$ and $\frac{b(n)}{a(n)} \to 0$. 
Without loss of generality assume that $\underset{n}{\text{sup }} a(n),\ b(n) \le 1$.
The sequences $\{M^1 _{n}\}_{n \ge 1}$
and $\{M ^2 _{n} \}_{n \ge 1}$ are suitable martingale difference
noise terms.
\paragraph{}
Suppose there exists $x_0 \in \mathbb{R}^d$ such that $g(x_0) \ge 0$, then $\mu = (\infty, \dots, \infty)$
maximizes $f(x_0) + \mu ^{\mathsmaller T} g(x_0)$. 
With respect to the faster timescale ($\mu$) iterates the slower timescale ($x$) iterates
can be viewed as being ``quasi-static'', see \cite{BorkarBook} for more details.
It then follows from the aforementioned observation that the 
$\mu$ iterates cannot be guaranteed to be stable. In other words, we cannot use 
(\ref{primalSA}) to solve the primal problem.
\paragraph{}
If strong duality holds then solving (\ref{primal}) is equivalent to solving its dual
given by:
\begin{equation}\label{dual}
 \underset{\substack{\mu \in \mathbb{R}^k \\ \mu \ge 0}}{sup} \ 
 \underset{x \in \mathbb{R}^d}{inf} 
  \left( f(x) + \mu ^{\mathsmaller T} g(x) \right) . 
\end{equation}
Further, the two timescale scheme to solve the dual problem is given by:
\begin{equation}\label{dualSA}
\begin{split}
x_{n+1} = x_n - a(n) \left[ \nabla_{x} \left( f(x_n) + \mu_n ^{\mathsmaller T} g(x_n) \right) + M^2 _{n+1} \right], \\
\mu_{n+1} = \mu_{n} + b(n) \left[ \nabla_{\mu}  \left( f(x_n) + \mu_n ^{\mathsmaller T} g(x_n) \right) + M^1 _{n+1} \right].
 \end{split}
\end{equation}
Note that (\ref{dualSA}) is obtained by flipping the timescales of (\ref{primalSA}). Strong
duality can be enforced if we assume the following:
\begin{enumerate}
 \item[(S1)] $f(x) = x^{\mathsmaller T} Q x + b^{\mathsmaller T}x + c$, where $Q$ is a positive
 semi-definite $d \times d$ matrix, $b \in \mathbb{R}^d$ and $c \in \mathbb{R}$.
 \item[(S2)] $g = A$, where $A$ is a $k \times d$ matrix.
 \item[(S3)] $f$ is bounded from below.
\end{enumerate}
The reader is referred to
Bertsekas \cite{bertsekas} for further details. 
For the purposes of this section we assume the following:
\begin{itemize}
 \item[] $(S1)-(S3)$ are satisfied.
 \item[$(A3)'$] $\sum_{n \ge 0} a(n) M^i _{n+1} < \infty$ a.s., where $i = 1,\ 2$.
\end{itemize}
The sole purpose of $(A3)$ in Section~\ref{assumptions} is to ensure the
convergence of the martingale noise terms \textit{i.e.}, $(A3)'$ holds.
It is clear that (\ref{dualSA}) satisfies $(A1)$ since $(S1)-(S3)$ hold 
while $(A2)$ is the step size assumption that is enforced.
\paragraph{}
The stability of the $\mu$ iterates in (\ref{dualSA}) directly 
follows from strong duality and $(A3)'$.
The $\mu$ iterates are ``quasi-static'' with respect to
the $x$ iterates. Further, 
since $f(x) + \mu _0 ^{\mathsmaller T} g(x)$ is a 
convex function (from $(S1)$ and $(S2)$),
for a fixed $\mu _0$, $f(x) + \mu _0 ^{\mathsmaller T} g(x)$ achieves its
minimum ``inside'' $\mathbb{R}^d$. Hence, the stability of the $x$ iterates will follow from
that of the $\mu$ iterates and $(A3)'$.
In other words, 
(\ref{dualSA}) satisfies $(A1),\ (A2), \ (A3)' \ \& \ (A4)$, see 
Section~\ref{assumptions} for the definitions of $(A1), (A2)$ and $(A4)$.
\paragraph{}
For a fixed $\mu _0$, the minimizers of $f(x) + \mu _0 ^{\mathsmaller T} g(x)$
constitute the global attractor of the o.d.e., 
$\dot{x}(t) = - \nabla_x (f(x) + \mu _0 ^{\mathsmaller T} g(x))$. 
Our paradigm comes in handy when this attractor set is \textit{NOT} \textit{singleton}, which is 
generally the case. 
In other words, we can define the following set
valued map: $\lambda _m : \mathbb{R}^k \to \mathbb{R}^d$, where 
$\lambda _m (\mu_0)$  is the global attractor of
$\dot{x}(t) = - \nabla_x (f(x) + \mu _0 ^{\mathsmaller T} g(x))$.
\paragraph{}
Now we check that (\ref{dualSA}) satisfies $(A5)$. To do so 
it is enough to ensure that $\lambda _m$ is an
upper semi-continuous map. Recall that $\lambda _m (\mu)$ is the minimum set
of $f(x) + \mu ^{\mathsmaller T} g(x)$ for each $\mu \in \mathbb{R}^k$. Dantzig,
Folkman and Shapiro \cite{dantzig} studied the continuity of minimum sets
of continuous functions. 
A wealth of sufficient conditions can be found in \cite{dantzig}
which when satisfied by the functions 
guarantee ``continuity'' of the corresponding minimum sets. 
In our case since $(S1)-(S3)$ are
satisfied, \textit{Corollary I.2.3} of \cite{dantzig} guarantees upper
semi-continuity of $\lambda _m$. 
\paragraph{}
Since (A1)-(A5) are satisfied by (\ref{dualSA}), it follows from 
Theorems~\ref{fastx} \&~\ref{slowy} that:
\begin{itemize}
 \item[(I)] Almost surely every limit point of $\{ \overline{x}(r + \cdotp) \ |\ r \ge 0\}$ in 
 $C([0, \infty), \mathbb{R}^d)$ is of the form $x(t) = x(0) + \int _0 ^t \nabla_x
 (f(x(t)) + \mu_0 ^{\mathsmaller T} g(x(t))) \,dt$ for some $x(0) \in \mathbb{R}^d $
 and some $\mu_ 0 \in \mathbb{R}^k$. 
 \item[(II)] Almost surely, any limit point of 
 $\{\widetilde{\mu}(r + \cdotp) \ | \ r \ge 0\}$ in $C([0, \infty), \mathbb{R}^k)$ is of
 the form $\mu (t) = \mu (0) + \int_0 ^t \nu (z) \,dz$ for some measurable function
 $\nu$ with $\nu (t) \in G(\mu (t))$, $t \ge 0$ and 
 $G(\mu (t)) = 
 \overline{co} \left( \{ \nabla_{\mu}(f(x) + \mu (t) ^{\mathsmaller T} g(x)) \ | \ x \in \lambda _m (\mu (t)) \} \right)$. 
\end{itemize}
For the construction of $\overline{x}(\cdotp)$ and 
$\widetilde{\mu}(\cdotp)$ see equations (\ref{xbar}) and (\ref{ytilde}) respectively.
If in addition, (\ref{dualSA}) satisfies $(A6)$ 
\textit{i.e.,} $\exists \ A_{\mu} \subset \mathbb{R}^k$ such that it is the global
attractor of $\mu (t) \in G(\mu (t))$,
then it follows from Theorem~\ref{main} that: 
almost surely any accumulation point of $\{(x_n, y_n) \ |\ n \ge 0\}$
belongs to the set 
$\mathcal{A} := \underset{\mu \in A_m}{\cup} \{ (x, \mu) \ | \ x \in \lambda_m (\mu) \}$.
The attractor $A_{\mu}$ is the maximum set of $H(\mu) :=$
$\underset{x \in \mathbb{R}^d}{inf} \left( f(x) + \mu ^{\mathsmaller T} g(x) \right)$
subject to $\mu \ge 0$. 
It may be noted that $H$ is a concave function that is bounded above as a consequence
of strong duality. For any $(x^*, \mu^*) \in \mathcal{A}$ 
we have that
\begin{equation} \nonumber
 f(x^*) + (\mu^*)^{\mathsmaller T}g(x^*) = \underset{\substack{\mu \in \mathbb{R}^k \\ \mu \ge 0}}{sup}
 \\ \underset{x \in \mathbb{R}^d}{inf} f(x) + \mu^{\mathsmaller T}g(x).
\end{equation}
In other words, almost surely the two timescale iterates given by (\ref{dualSA})
converge to a solution of the dual (\ref{dual}). It follows from strong duality that
they almost surely converge to a solution of the primal (\ref{primal}).
\section{Conclusions}
In this paper we have presented a framework for the analysis of two timescale stochastic
approximation algorithms with set valued mean fields. Our framework generalizes the one by Perkins and Leslie. 
We note that the analysis of the faster timescale proceeds in a predictable manner
but the analysis of the slower timescale is new to the literature
to the best of our knowledge.
As an application we analyze the two timescale scheme that arises from 
the Lagrangian dual problem in optimization using our framework. Our framework is applicable even when
the minimum sets are not singleton. 

\bibliographystyle{plain}
\bibliography{TTforDI}

\end{document}